\documentclass[journal,transmag]{IEEEtran}
\usepackage[comma,numbers,square,sort&compress]{natbib}

\usepackage{amsmath}
\usepackage{amssymb}
\usepackage{float}
\usepackage{amsthm}
\usepackage{graphicx}
\usepackage{epstopdf}
\usepackage{color}
\usepackage{verbatim}
\usepackage{tikz,times}
\usepackage{algorithm}
\usepackage{algorithmic}
\newtheorem{dingyi}{Definition}[section]

\newtheorem{dingli}{Theorem}[section]

\newtheorem{tuilun}{Corollary}
\newtheorem{zhushi}{Remark}[section]

\newtheorem{yinli}{Lemma}[section]

\newcommand{\dkh}[1]{\left(#1\right)}
\newcommand{\hkh}[1]{\left\{#1\right\}}
\newcommand{\fkh}[1]{\left[#1\right]}
\allowdisplaybreaks[4]
\begin{document}

	\title{Asymptotic Properties of Distributed Social Sampling Algorithm}
	
	\author{
		Qian Liu,
		Xingkang~He,
 Haitao Fang
		
		\thanks{Qian Liu and Haitao Fang  are with  LSC, Academy of Mathematics and Systems Science, Chinese Academy of Sciences, Beijing 100190, China; they are also with School of Mathematical Sciences, University of Chinese Academy of Sciences, Beijing 100049, China.}
			\thanks{Xingkang He is with the  Department of Automatic Control,  KTH Royal Institute of
			Technology, SE-100 44 Stockholm, Sweden. (xingkang@kth.se)}

	}
	
	
	\maketitle


	\begin{abstract}
Social sampling is a novel randomized message passing protocol inspired by social communication  for  opinion formation in social networks. In a typical social sampling algorithm, each agent holds a sample from the empirical distribution of social opinions at initial time, 
and it collaborates with other agents in a distributed manner to estimate the initial empirical distribution by randomly sampling a message from current distribution estimate.
In this paper, we focus on analyzing the theoretical properties of the distributed social sampling algorithm over random networks. Firstly, we provide a framework based on stochastic approximation to study the asymptotic properties of the algorithm. Then, under mild conditions, we prove that the estimates of all agents converge to a common random distribution, which is composed of the initial empirical distribution and the accumulation of quantized error.  Besides, by tuning algorithm parameters, we prove the strong consistency, namely, the distribution estimates of agents almost surely converge to the initial empirical distribution. 
Furthermore, the asymptotic normality of estimation error generated by distributed social sampling algorithm is addressed. Finally, we provide a numerical simulation to validate the theoretical results of this paper.
	\end{abstract}

	\begin{IEEEkeywords}
		Social networks; opinion formation; social sampling; stochastic approximation; random networks; asymptotic normality
	\end{IEEEkeywords}

	\section{Introduction}
	In social networks,  the study of opinion formation is to model the  fragmentation or merging of opinions among agents in a society.  A large class of real world phenomena can be well interpreted with  opinion dynamics, such as election forecasting \cite{Holley1975Ergodic}, analysis of public opinions \cite{Acemoglu2011Bayesian}, language evolution \cite{Narayanan2014Language} and so on. 
	During the past decades, with the development of network technology and increasing of social communication,
	more and more researchers are focusing on the study of opinion formation in social networks as well as new approaches to distributed learning and estimation \cite{Friedkin2016Network,Hegselmann2002Opinion,Zhang2013Opinion,Pineda2009Noisy,Boyd2006Randomized,Lou2017Reference,Frasca2015Distributed}. \\
	\indent The  results in sociology on opinion formation are mainly based on empirical studies, which 
	usually lack sufficient theory foundation.  Thus, in recent years, the requirements for mathematical modeling and theoretical analysis are increasing. 
	Generally, in opinion dynamics, the communications between agents  can be modeled by a network or a graph, whose edges represent channels through which one agent can share information with others. Naturally, the following issues are concerned: whether an equilibrium or consensus can be achieved via such social interactions? If the answer is yes, is this equilibrium influenced by  network topology,  initial opinions of agents or  communication protocol?  Otherwise, how can the opinion dynamics behave?                                              
	There are quite a few papers considering the above problems. The early work \cite{Friedkin1999Social} proposes a common paradigm called Friedkin and Johnson model, where agents are  divided into  stubborn agents and regular agents. In \cite{Ravazzi2015Ergodic},  the Friedkin and Johnson model is viewed as a randomized gossip algorithm inducing oscillations, which should be ergodic under some stable assumptions.  Besides, a tractable communication model is developed in \cite{doi:10.3982/TE1204} to study the dynamics of belief formation and information aggregation. 	The authors propose  asymptotic learning to describe that the fraction of agents can behave properly, and provide sufficient and necessary conditions to guarantee the asymptotic learning. A continuous-time opinion dynamic model with stochastic gossip process is proposed in \cite{doi:10.1287/moor.1120.0570} to investigate the generation of disagreement and fluctuation. It is shown that the society containing stubborn agents with different opinions  keeps fluctuating in an ergodic manner.\\
	\indent Most of existing results focus on modeling the dynamics of consensus, diversity or fluctuation in social networks, where opinions are represented as scaler variables. A nature extension is that how the agents can obtain a common knowledge of the global phenomenon. In what follows, we discuss the local reconstruction of the empirical distribution of initial opinions via social interactions. In fact, probability distribution of  opinions is a proper way to model the enormous opinions of agents on some certain subjects in a complex society, such as the election candidates they support or prefer.
	The aim of the agents is to estimate the discrete empirical distribution derived by the average of initial opinions through local interactions. This algorithm is essentially a randomized approximation of consensus procedure. So far, there have been a wide range of work on networked consensus in terms of communication noise, time delay, network topology and so on\cite{Degroot1974Reaching,Borkar1982Asymptotic,Tsitsiklis1983Convergence,Olfatisaber2004Consensus,Imai2009Distributed,6481629}. Nevertheless, due to the limitation of communication cost, 
	each agent cannot exchange their entire opinions completely.
	To deal with this problem, a communication scheme called social sampling is proposed \cite{6827923}.
	The idea of this scheme is that each agent can only share a sample generated randomly by following its current distribution. The distributed social sampling algorithm given in \cite{6827923} is a randomized approximation of consensus procedures, in which a group of agents aim to reach a common decision in a distributed way.  Since the transferred message is quantized as an identical vector, the computation complexity is  significantly reduced.
	However, the theoretical properties, such as the effects of network topology or the quantized process, have not been well investigated. Besides, one interesting theoretical problem is the asymptotic normality of estimation error.  The work  \cite{Rajagopal2008Network} considers such problem for certain cases of noisy communication in consensus 
	schemes for scalars and finds a connection between network topology and covariance matrix of the error limitation distribution. However, such a result cannot be transfered to distributed social sampling algorithm because of the quantized error term. In this paper, we will complete the analysis and establish asymptotic normality for the social sampling scenario.\\
	\indent The main contributions of this paper are threefold:
	\begin{description}
		\item[(i)] We provide a novel analysis framework based on stochastic approximation to study the asymptotic properties of the distributed social sampling algorithm over random networks. To ensure the convergence of the algorithm in an almost sure sense, we use the techniques of stochastic approximation \cite{Fang2012ON}, in which state space is decomposed into two parts:  consensus part and  vanishing part.  Besides, some  analysis methods provided in this paper can contribute to further related researches.	
		\item[(ii)] For consensus over random networks, the strong consensus is most desirable. Compared with \cite{6827923}, this is achieved in our work under the milder conditions on network structures and communication noise properties. We prove that the distribution estimates of agents reach consensus almost surely to the value related with true empirical distribution and the accumulation of quantized error. Besides, by tuning algorithm parameters, we prove the strong consistency, $i.e.$,  the distribution estimates of agents are almost surely convergent to the initial empirical distribution. On the other hand, unlike the fixed topology used in \cite{6827923}, the condition on network topology is relaxed to joint connectivity of mean digraphs for random networks. 
		\item[(iii)] We provide convergence rate of estimation of the social sampling protocol. Explicitly, we prove that the overall estimation errors of the algorithm  are asymptotically normal with zero mean and known covariance matrix. 
		The covariance matrix shows that how networks and quantized error influence the estimation performance of the distributed social sampling algorithm. Compared with \cite{Rajagopal2008Network}, the conditions on communication noise in this work are more general.
	\end{description}
	
	\par
	The reminder of the paper is organized as follows. Section {\bf \ref{Prelim}} provides some preliminary information about graph theory and the main problem we considered. In section {\bf \ref{dierjie}}, we describe the social sampling protocol and the stochastic algorithm studied in this paper. Convergence analysis for the distributed algorithm is given in section {\bf \ref{disanjie}}, while asymptotic properties are presented in section {\bf \ref{disijie}}.  Section {\bf \ref{diwujie}} shows a numerical simulation. Section {\bf \ref{diliujie}} gives some concluding remarks.
\subsection{Notations}
Let $\mathbf{e}_{i}\in \mathbb{R}^{M}$ stand for the  unit row vector whose $i$-th element equals to 1.  $\mathbf{I}_{\fkh{\cdot}}$ denotes the indicator function and $\mathbf{1}$ stands for the proper dimensional column vector with elements all being 1. $\mathbf{I}_{N}$ denotes the $N$-dimension identity matrix. The superscript `` T '' represents the transpose.  The abbreviation $i.i.d.$ stands for independent identically distribution of random variables,  $N\dkh{0,S }$ denotes the normal distribution with zero mean and covariance $S$. $E\fkh{x}$ denotes the mathematical expectation of the stochastic variable $x$. Notation $diag\dkh{\cdot}$ represents the diagonalization of scalar elements.
$\mathbb{R}^{n}$ represents the $n$-dimension Euclidean space. Besides, we denote $\fkh{N}=\hkh{1,2,\cdots,N}.$	
\section{Preliminaries}\label{Prelim}
\subsection{Graph Theory }
Let $\mathcal{G}=\dkh{\mathcal{V},\mathcal{E},\mathbf{W}}$ be a weighted digraph, where $\mathcal{V}=\fkh{N}$ is label set of $N$ agents, $\mathcal{E}\subset \mathcal{V}\times\mathcal{V}$ is the edge set, where an ordered pair $\dkh{i,j}\in\mathcal{E}$ means that agent $i$ can get information from agent $j$ directly. The {\it in-neighbor} set is denoted by $\mathcal{N}_{in}\dkh{i}=\hkh{j\in \mathcal{V}|\dkh{i,j}\in \mathcal{E}}$, while the {\it out-neighbor} set of agent $i$ is denoted by $\mathcal{N}_{out}\dkh{i}=\hkh{j\in \mathcal{V}|\dkh{j,i}\in \mathcal{E}}$.  The graph is undirected if it is bidirectional, $i.e.$ $\dkh{j,i}\in\mathcal{E}$ if and only if $\dkh{i,j}\in\mathcal{E}$. $\mathbf{W}=\fkh{w_{ij}}\in \mathbb{R}^{N\times N}$ is the weighted adjacency matrix of $\mathcal{G}$, where $w_{ij}>0$ if $\dkh{i,j}\in \mathcal{E}$, and $w_{ij}=0$ otherwise. The nonnegative matrix $\mathbf{W}$ is called row-wise stochastic if $\mathbf{W1}=\mathbf{1}$, and is called column-wise stochastic if $\mathbf{W}^{T}\mathbf{1}=\mathbf{1}$. We say $\mathbf{W}$ is double stochastic if it is both row-wise stochastic and column-wise stochastic.
\par
For any $i,j\in\fkh{N}$, the in-degree of agent $i$ is defined as $deg_{in}\dkh{i}=\sum_{j=1}^{N}w_{ij}$ and the out-degree of agent $i$ is defined as $deg_{out}\dkh{i}=\sum_{j=1}^{N}w_{ji}$. We say $\mathcal{G}$ is a balanced digraph,  if $deg_{in}\dkh{i}=deg_{out}\dkh{i}$ for any $i\in\fkh{N}$. The digraph $\mathcal{G}$ is strongly connected if for any pair $i,j\in \mathcal{V}$, there exists a  directed sequence of nodes $i_{1},i_{2},\dots,i_{p}\in \mathcal{V}$,  such that $\dkh{i,i_{1}}\in\mathcal{E}$,$\dkh{i_{1},i_{2}}\in\mathcal{E}$,$\dots$,$\dkh{i_{p},j}\in\mathcal{E}$. 
\par 
The network topology in this work is allowed to be time-varying, thus the weighted communication network at time $k$ is denoted by $\mathcal{G}_{k}=\dkh{\mathcal{V},\mathcal{E}_{k}, \mathbf{W}_{k}}$. The graph sequence $\hkh{\mathcal{G}_{k}}$ is called jointly connected, if there exists an integer $T>0$, such that $\dkh{\mathcal{V},\bigcup_{s=0}^{T}\mathcal{E}_{s}}$ is strongly connected.
\par 
Besides, we introduce a definition to characterize the asymptotic behavior of the agents.
\begin{dingyi}[Strong consensus \cite{huangminyi}]
	The estimates of agents $\dkh{i.e.~	Q_{i,k}}$ are said to reach strong consensus if there exists a random variable $\mathbf{q}^{\ast}$	such that,  with probability 1 and  for all $i\in\fkh{N}$, $\lim\limits_{k\rightarrow\infty}Q_{i,k}=\mathbf{q}^{\ast}$.  We also say that the estimates converge almost surely $\dkh{a.s.}$.
\end{dingyi}
\subsection{Distributed Learning of Distributions}  
Consider a network of $N$ agents. The communication relationship among agents is described by a sequence of directed graphs $\hkh{\mathcal{G}_{k}}$, where time is discrete and indexed by $k=\hkh{0,1,2,\dots}$. At initial time $k=0$,  every agent has a single discrete sample $X_{i}$ taking values in a finite state set $\chi=\fkh{M}$.  The problem we consider here is that agents in a network need to  learn the  empirical distribution $\it\Pi \triangleq \{{\it \Pi }\dkh{x},x \in \fkh{M}\}$, where
\begin{equation}\label{true distribution}
{\it \Pi }\dkh{x}=\frac{1}{N} \sum_{i=1}^{N} \mathbf{1}_{\fkh{X_{i}=x} }\mathbf{e}_{x},  \qquad   \forall x \in \fkh{M}=\hkh{1,2,\dots,M}.   
\end{equation}
It can be considered as the histogram of the initial distribution of the agent opinions over the network.
\par
For illustration, we consider a motivating example. A group of customers who want to buy  one product from $M$ competing new products.  Assume that every customer has an opinion or preference over each product, which can be quantized as scores. To get enough information, customers will communicate with their friends or neighbors about their opinions.  In the traditional message protocol \cite{Hegselmann2002Opinion,Zhang2013Opinion,Boyd2006Randomized}, the agents exchange their entire opinion histogram each time, which means customers will discuss the evaluations of every single product. This is not realistic, especially under the situation that there exist an enormous variety of products. A sample generated from the current estimate is transferred in social sampling protocol, which means customer simplifies the communication process and exchanges information about one kind of randomly selected product. The  analysis of this paper shows that the customer still could get enough information about all products. 

\section{Problem Setup}\label{dierjie}
\subsection{Algorithm Formulation}
In this paper, we aim to estimate this histogram through a randomized algorithm called {\it social sampling}\cite{6827923}. The algorithm is based on the sample generated from the current estimate $Q_{i,k}$ of the true distribution ${\it \Pi}$. At time $k$, each agent holds an internal estimate $Q_{i,k}$ of ${\it \Pi}$  with  $Q_{i,0}=\mathbf{e}_{X_{i}}$.  We treat $Q_{i,k}$ as a probability distribution of the elementary vectors $\hkh{\mathbf{e}_{m}:m\in \fkh{M}}$, so that $Q_{i,k}$ should be probability vector on $\chi=\fkh{M}$. Agent $i$ generates its message social sample $Y_{i,k}$ as a function of the internal estimate $Q_{i,k}$, then agent $i$ sends $Y_{i,k}$ to its out-neighbors and receives the in-neighbor messages $\hkh{Y_{j,k}:j\in \mathcal{N}_{i}}$. At each iteration, agent $i$ uses the samples from its neighbors and current estimate to obtain the updated estimate $Q_{i,k+1}$.
\par
We assume $Y_{i,k}\in\mathcal{Y}=\hkh{\mathbf{e}_{1},\dots, \mathbf{e}_{M}}$, which can be viewed as a  label of the opinion state space. So the opinion takes values from a finite, discrete value space. The random message $Y_{i,k}\in \mathcal{Y}$ of agent $i$ at time $k$ is  generated according to the distribution $P_{i,k}\in \mathbb{P}\dkh{\mathcal{Y}}$, which  is a function of the internal estimate $Q_{i,k}$. More precisely, $P_{i,k}$ is a $M$-dimension row probability vector where the $m$-th element $P_{i,k}^{m}=\mathbb{P}\dkh{Y_{i,k}=\mathbf{e}_{m}}$.
\begin{zhushi}{\rm 
		We can choose $P_{i,k}$ properly and make it be a correction term associated to the internal estimate $Q_{i,k}$. For example, we set $P_{i,k}=0$ when $Q_{i,k}<\alpha$, where $\alpha$ is a presetting bound. Under some complicated situations,  such as the opinion space being extremely large or the histogram being far from uniform, this kind of censoring can avoid inefficient communication. Of course, 
		we can also choose $P_{i,k}=Q_{i,k}$ in some simple situations.}
\end{zhushi}
Next, we will formulate the distributed social sampling algorithm and write it in a compact form. For notational convenience,  the social samples at time $k$ are denoted by a $NM$-dimension column vector $\mathbf{Y}_{k}\triangleq \dkh{Y_{1,k}^{T},Y_{2,k}^{T},\dots,Y_{N,k}^{T}}^{T}\in \mathbb{R}^{NM}$, which is generated from the sampling function $\mathbf{P}_{k}\triangleq \dkh{P_{1,k}^{T},P_{2,k}^{T},\dots,P_{N,k}^{T}}^{T}\in \mathbb{R}^{NM}$. Similarly, we set $\mathbf{Q}_{k}\triangleq \dkh{Q_{1,k}^{T},Q_{2,k}^{T},\dots,Q_{N,k}^{T}}^{T}\in \mathbb{R}^{NM}$.\\
\indent For agent $i$ at time $k$, the internal estimate $Q_{i,k}$  is updated in a distributed way as follows:
\begin{align}\label{linear update}
Q_{i,k+1}=&\dkh{1-\delta_{k}a_{ii}^{k}} Q_{i,k}- \delta_{k} b_{ii}^{k} Y_{i,k}\nonumber\\
& +\sum _{j\in \mathcal{N}_{i}\dkh{k}} \delta_{k}w_{ij}^{k} Y_{j,k},
\end{align}
where $a_{ii}^{k}$, $b_{ii}^{k}$ are  communication coefficients subject to $a_{ii}^{k}\geq 0$, $b_{ii}^{k}\geq 0$. $\mathbf{W}_k=\fkh{w_{ij}^{k}}\in \mathbb{R}^{N\times N}$ is the weighted adjacency matrix of the network topology and $\delta_{k}$ is the step size. By designing update procedure like this, we can add some reasonable assumptions on the coefficients to  guarantee that the internal estimate $Q_{i,k+1}$ is a probability vector on the opinion state space $\chi=\fkh{M}$ at any time for every $i\in \fkh{N}$.\\
\indent This paper focuses on solving the following two problems: {\it i) } analyze the conditions ensuring  the convergence of distributed social sampling algorithm $\dkh{\ref{linear update}}$ over random networks in the almost sure sense.
{\it ii) } derive the asymptotic normality of algorithm $\dkh{\ref{linear update}}$ and characterize the effect of random sampling protocol and network topology on the limit  covariance matrix.


\section{Consensus and Consistency}\label{disanjie}
In this section, we  provide  an analysis framework based on stochastic approximation to study the convergence of $\dkh{\ref{linear update}}$. \\
\indent Denote $\mathbf{A}_{k} \triangleq diag\dkh{a_{11}^{k},\cdots,a_{NN}^{k}}$, $\mathbf{B}_{k}\triangleq diag\dkh{b_{11}^{k},\cdots,b_{NN}^{k}}$,
then we can write $\dkh{\ref{linear update}}$ in a  compact form
\begin{align}\label{stochastic iteration}
\mathbf{Q}_{k+1}=\mathbf{Q}_{k}+&\delta_{k}\bigg\{\dkh{ \dkh{\mathbf{W}_{k}-\mathbf{B}_{k}-\mathbf{A}_{k}}\otimes\mathbf{I}_{M}}\mathbf{Q}_{k}\nonumber\\
&+\dkh{ \dkh{\mathbf{W}_{k}-\mathbf{B}_{k}}\otimes\mathbf{I}_{M}}\dkh{\mathbf{Y}_{k}-\mathbf{Q}_{k}}\bigg\},
\end{align} 
where ``$\otimes$" is the Kronecker product.
\par
Suppose that the $\sigma$-algebra  $\mathcal{F}_{k}\triangleq \sigma\hkh{Q_{i,0},\mathbf{W}_{t}, \mathbf{B}_{t}, 1\leq i \leq N,  0\leq t\leq k}$
is a filtration of the basic  probability space $\dkh{\Omega,\mathcal{F},\mathbb{P}}$. Hence $\mathbf{Q}_{k}$ is measurable with respect to $\mathcal{F}_{k}$. Given the update rule in $\dkh{\ref{stochastic iteration}}$, the consensus of the opinion dynamics is equivalent to the convergence of linear stochastic approximation algorithms. The linear matrix $\dkh{\mathbf{W}_{k}-\mathbf{B}_{k}-\mathbf{A}_{k}}\otimes\mathbf{I}_{M}$ represents the effect of network topology at time $k$, while $\dkh{\dkh{\mathbf{W}_{k}-\mathbf{B}_{k}}\otimes\mathbf{I}_{M}}\dkh{\mathbf{Y}_{k}-\mathbf{Q}_{k}}$ is the error item caused by the quantized data. 
\par
As shown in iteration $\dkh{\ref{stochastic iteration}}$,  the opinion formation process can be considered as a linear regression case of stochastic approximation. Next, we will analyze $\dkh{\ref{stochastic iteration}}$ with stochastic approximation. To begin with, the following assumptions are given.
\begin{description}
	\item[ $\mathbf{A1}$] $\delta_{k}\xrightarrow [k\rightarrow \infty]{}0$, $\delta_{k}>0$, $\sum_{k=0}^{\infty}\delta_{k}=\infty$,  $\sum_{k=0}^{\infty}\delta_{k}^{2}<\infty$ and $\frac{1}{\delta_{k+1}}-\frac{1}{\delta_{k}}\xrightarrow[k\rightarrow \infty]{}\delta\geq 0$.
	\item[ $\mathbf{A2}$] 	
	\begin{enumerate}
		\item[$\dkh{i}$] $\hkh{\mathbf{W}_{k}}_{k\geq 0}$ is an independent random sequence with expectation denoted by $\bar{\mathbf{W}}_{k}=E\fkh{\mathbf{W}_{k}}$ and the adjacency matrix  $\mathbf{W}_{k}=\fkh{w_{ij}^{k}}$ is double stochastic.	Besides, there exists an uniform bound $\bar{w}_{ij}^{k}>\tau>0$, $\forall k>0$ for all nonzero $\bar{w}_{ij}^{k}\neq 0$.
		\item[$\dkh{ii}$] There is an integer $T>0$,  such that the mean graph $\bar{\mathcal{G}}_{k}=\dkh{\mathcal{V},\mathcal{E}_{k},\bar{\mathbf{W}}_{k}}$ generated by $\hkh{\bar{\mathbf{W}}_{k}}$ is jointly connected in the fixed period $\fkh{k,k+T}$, $i.e.$ there exits an integer $T>0$, such that the graph  $\dkh{\mathcal{V},\bigcup_{s=0}^{T}E\{\mathcal{E}_{k+s}\}}$ is strongly connected.
	\end{enumerate}
\end{description}
\par
In addition, we need extra assumptions on the mixed coefficients and the social sampling protocol.
\begin{description}
	\item[ $\mathbf{A3}$] $\|\mathbf{P}_{k}-\mathbf{Q}_{k}\|^{2}\xrightarrow[k\rightarrow \infty]{}0,~a.s.$ .
	\item[ $\mathbf{A4}$] The communication coefficients $a_{ii}^{k}$ and $b_{ii}^{k}$ are chosen properly such that $a_{ii}^{k}+b_{ii}^{k}=1$ for any $k\geq 0$, $i.e.$,  $\mathbf{A}_{k}+\mathbf{B}_{k}=\mathbf{I}_{N}$.
\end{description}
\par
For convenience of analysis, we arrange  the algorithm $\dkh{\ref{stochastic iteration}}$ as follows.
\begin{align}
\mathbf{Q}_{k+1}=&\mathbf{Q}_{k}+\delta_{k}\bigg\{\dkh{ \dkh{\mathbf{W}_{k}-\mathbf{I}_{N}}\otimes\mathbf{I}_{M}}\mathbf{Q}_{k}\nonumber\\
&+\dkh{ \dkh{\mathbf{W}_{k}-\mathbf{B}_{k}}\otimes\mathbf{I}_{M}}\dkh{\mathbf{Y}_{k}-\mathbf{Q}_{k}}\bigg\}\nonumber\\
=&\mathbf{Q}_{k}+\delta_{k}\bigg\{\dkh{ \dkh{\bar{\mathbf{W}}_{k}-\mathbf{I}_{N}}\otimes\mathbf{I}_{M}}\mathbf{Q}_{k}\nonumber\\
&+\dkh{ \dkh{\bar{\mathbf{W}}_{k}-\mathbf{B}_{k}}\otimes\mathbf{I}_{M}}\dkh{\mathbf{P}_{k}-\mathbf{Q}_{k}}\nonumber\nonumber\\
&+\dkh{ \dkh{\mathbf{W}_{k}-\bar{\mathbf{W}}_{k}}\otimes\mathbf{I}_{M}}\mathbf{P}_{k}\nonumber\\
&+\dkh{ \dkh{\mathbf{W}_{k}-\mathbf{B}_{k}}\otimes\mathbf{I}_{M}}\dkh{\mathbf{Y}_{k}-\mathbf{P}_{k}}\bigg\}.
\end{align}
\par
Denote 
\begin{align}
\left\{
\begin{array}{l}
\bar{\mathbf{H}}_{k}\triangleq \dkh{\bar{\mathbf{W}}_{k}-\mathbf{I}_{N}}\otimes\mathbf{I}_{M},\\
\mathbf{C}_{k}\triangleq \dkh{\dkh{\bar{\mathbf{W}}_{k}-\mathbf{B}_{k}}\otimes\mathbf{I}_{M}}\dkh{\mathbf{P}_{k}-\mathbf{Q}_{k}},\\
\mathbf{M}_{k}\triangleq \dkh{\dkh{\mathbf{W}_{k}-\mathbf{B}_{k}}\otimes\mathbf{I}_{M}}\dkh{\mathbf{Y}_{k}-\mathbf{P}_{k}}\\
\quad\qquad+\dkh{ \dkh{\mathbf{W}_{k}-\bar{\mathbf{W}}_{k}}\otimes   \mathbf{I}_{M}}\mathbf{P}_{k}, 
\end{array}
\right.\label{yangchalie}
\end{align}
then we have 
\begin{equation}\label{stochastic iteration 2}
\mathbf{Q}_{k+1}=\mathbf{Q}_{k}+ \delta_{k}\dkh{\bar{\mathbf{H}}_{k}\mathbf{Q}_{k}+\mathbf{C}_{k}+\mathbf{M}_{k}}.
\end{equation}

\begin{zhushi}
	{\rm 	Condition $\mathbf{A1}$ can be automatically satisfied if $\delta_{k}=\frac{a}{k^{\delta}}$ with $a>0$, $\delta\in \left( \frac{1}{2}, 1\right]$. In fact, we can pick $\mathbf{P}_{k}=\mathbf{Q}_{k}$, $i.e.$,  we generate social sample from internal estimate directly without censoring, which means   $\mathbf{C}_{k}=0$. The double stochastic assumption on $\bar{\mathbf{W}}_{k}$ means that the mean graph should be balanced.  The  lower bound $\tau$ in $\mathbf{A2}$ for the nonzero elements is used to guarantee the stability of linear matrix sequence  $\hkh{\mathbf{H}_{k}}$, which is easily satisfied in the case where the network is switched over a finite number of network topologies.
	}
\end{zhushi}

\par
Before presenting the consensus results for the  algorithm $\dkh{\ref{stochastic iteration 2}}$, we provide the  following lemma.
\begin{yinli}[see \cite{Fang2012ON}]\label{fang}
	Let $\hkh{H_{t}}$ be $n\times n$-matrices with $\sup_{t}\|H_{t}\|<\infty$. Assume that there is an $n\times n$-matrix $U>0$ and an integer $K>0$ such that for $\forall\  t\geq 0,$ 
	\begin{align*}
	UH_{t}+H_{t}^{T}U\leq 0 \text { and } \sum_{s=t}^{t+K}\dkh{UH_{s}^{T}+H_{s}U}\leq -\beta I_n, \beta>0.
	\end{align*}
	If step-size  $\hkh{\delta_{k}}$ satisfies $\mathbf{A1}$ and  $\omega_{t}$ can be expressed as $\omega_{t}=\mu_{t}+\nu_{t}$ where
	\begin{align*}
	\sum_{t=1}^{\infty} \delta_{t}\mu_{t+1}<\infty \text{ and } \nu_{t}\xrightarrow[t\rightarrow \infty]{}0,~a.s..
	\end{align*}
	Then, for an arbitrary initial value $x_{0}$, the sequence $\hkh{x_{t}}$ generated by $x_{t+1}=x_{t}+\delta_{t}\dkh{H_{t}x_{t}+\omega_{t}}$
	converges to zero almost surely.
\end{yinli} 

\par 

Note that,  in expression \eqref{stochastic iteration 2}, we have separated  the quantized error into two parts: $\mathbf{C}_{k}$ is a censoring item associated  to the difference between the internal matrix $\mathbf{Q}_{k}$ and the sampling matrix $\mathbf{P}_{k}$,  and $\mathbf{M}_{k}$ is a martingale difference sequence, which will be demonstrated in the following lemma.
\begin{yinli}\label{proof}
	$\dkh{\mathbf{M}_{k}, \mathcal{F}_{k}}$ is a martingale difference sequence under $\mathbf{A2}$.   	
\end{yinli}
\begin{proof}
	See the proof in Appendix A.
\end{proof}
\par
As claimed above, for the opinion dynamic consensus we only need to show the convergence of $\mathbf{Q}_{k}$ given by $\dkh{\ref{stochastic iteration 2}}$ almost surely. This is given by the following theorem. It is shown that all estimates of agents will achieve consensus to a common estimate based on empirical distribution.
\begin{dingli}[\textbf{Consensus}]\label{theorem}
	Let $\hkh{\mathbf{Q}_{k}}$ be generated by the algorithm $\dkh{\ref{stochastic iteration 2}}$. Under the conditions $\mathbf{A1}$, $\mathbf{A2}$, $\mathbf{A3}$ and $\mathbf{A4}$, we have  
	\begin{equation}
	\lim\limits_{k\rightarrow \infty}\mathbf{Q}_{k} =\mathbf{Q}^{\ast}, ~~~a.s. ,
	\end{equation}
	where  $\mathbf{Q}^{\ast}=\mathbf{1}\otimes q^{\ast}$.  Explicitly,
	$\lim\limits_{k \rightarrow \infty}Q_{i,k}= q^{\ast}$, a.s.,
	where
	\begin{align*}
	&q^{\ast}\triangleq\dkh{\frac{1}{N}\mathbf{1}^{T}\otimes\mathbf{I}_{M}}\mathbf{Q}_{0}\\
	&+\frac{1}{N}\sum_{k=0}^{\infty}\delta_{k}\dkh{\mathbf{1}^{T}\dkh{\mathbf{W}_{k}-\mathbf{B}_{k}}\otimes\mathbf{I}_{M}}\dkh{\mathbf{Y}_{k}-\mathbf{Q}_{k}}, ~~\forall i \in \fkh{N}.
	\end{align*}	
	Furthermore, if    $\|\mathbf{P}_{k}-\mathbf{Q}_{k}\|=O\dkh{\delta_{k}}$,  then $q^{\ast}<\infty,~$ a.s..
\end{dingli}
\begin{proof}
	The  mixed-product property of Kronecker product, $\dkh{A\otimes B}\dkh{C\otimes D}=AC\otimes BD$, will be  frequently used in the following.
	\par
	Firstly, we write $\mathbf{Q}_{k}$ as a sum of a vector in the consensus space and a disagreement vector by orthogonal decomposition. Let  $T\triangleq\fkh{\begin{matrix}
		T_{1}\\\frac{1}{\sqrt{N}}\mathbf{1}^{T}
		\end{matrix}}$
	be an orthogonal matrix, then we have $T_{1} \mathbf{1}=\mathbf{0}$ and  $T_{1}T_{1}^{T}=\mathbf{I}_{N-1}$. Set ${\it\Gamma}\triangleq \mathbf{I}_{N}- \frac{1}{N} 
	\mathbf{1} \mathbf{1}^{T}$ ,  then $T {\it\Gamma} =\fkh{\begin{matrix}
		T_{1}\\\mathbf{0}
		\end{matrix}}$.
	Pre-multiplying  $\dkh{\ref{stochastic iteration}}$ by $T{\it\Gamma}\otimes \mathbf{I}_{M}$ yields 	
	\begin{align}
	&\dkh{\fkh{\begin{matrix}
			T_{1}\\\mathbf{0}
			\end{matrix}}\otimes \mathbf{I}_{M}}\mathbf{Q}_{k+1}\nonumber\\
		=& \dkh{\fkh{\begin{matrix}
			T_{1}\\\mathbf{0}
			\end{matrix}}\otimes \mathbf{I}_{M}}\mathbf{Q}_{k}+ \delta_{k}\bigg\{\dkh{\fkh{\begin{matrix}
			T_{1}\\\mathbf{0}
			\end{matrix}}\otimes \mathbf{I}_{M}}
	\bar{\mathbf{H}}_{k}\mathbf{Q}_{k}\nonumber\\
	&+\fkh{\dkh{\begin{matrix}
			T_{1}\\\mathbf{0}
			\end{matrix}}\otimes \mathbf{I}_{M}} \mathbf{C}_{k}+\dkh{\fkh{\begin{matrix}
			T_{1}\\\mathbf{0}
			\end{matrix}}\otimes \mathbf{I}_{M}}\mathbf{M}_{k}\bigg\}.
	\end{align}	
	Setting $\xi_{k} \triangleq \dkh{T_{1}\otimes \mathbf{I}_{M}}\mathbf{Q}_{k}$,  we obtain 
	\begin{align}\label{jiangweixiangliang}
	\dkh{T{\it\Gamma}\otimes \mathbf{I}_{M}}\mathbf{Q}_{k}=\fkh{\xi_{k}^{T},0}^{T},
	\end{align}
	and
	\begin{align}\label{jiangwei}
	\xi_{k+1}=&\xi_{k}+ \delta_{k}\hkh{\dkh{T_{1}\otimes \mathbf{I}_{M}}\{\dkh{\bar{\mathbf{W}}_{k}-\mathbf{I}_{N}}\otimes\mathbf{I}_{M}}\nonumber\\
		&\cdot\dkh{T_{1}^{T}\otimes \mathbf{I}_{M}}\xi_{k}+\dkh{T_{1}\otimes\mathbf{I}_{M}}\dkh{\mathbf{C}_{k}+\mathbf{M}_{k}}\}\nonumber\\
	=&\xi_{k}+ \delta_{k}\{\dkh{\dkh{T_{1}\dkh{\bar{\mathbf{W}}_{k}-\mathbf{I}_{N}}T_{1}^{T}}\otimes\mathbf{I}_{M}}\xi_{k}\nonumber\\
		&+\dkh{T_{1}\otimes\mathbf{I}_{M}}\mathbf{C}_{k}+\dkh{T_{1}\otimes\mathbf{I}_{M}}\mathbf{M}_{k}\}.
	\end{align} 
	\par
	Denote $\mathbf{F}_{k}\triangleq T_{1}\dkh{\bar{\mathbf{W}}_{k}-\mathbf{I}_{N}}T_{1}^{T}=T_{1}\bar{\mathbf{W}}_{k}T_{1}^{T}-\mathbf{I}_{N-1}$.  To use Lemma \ref{fang},  we need to verify the stability of matrix sequence  $\hkh{\mathbf{F}_{k}\otimes \mathbf{I}_{M}}$.  Since the adjacency matrix $\mathbf{W}_{k}$ is double stochastic, $i.e.$, $\mathbf{1}^{T}\mathbf{W}_{k}=\mathbf{1}^{T}$ and $\mathbf{W}_{k}\mathbf{1}=\mathbf{1}$, then 
	$\bar{\mathbf{W}}_{k}$ has the single largest eigenvalue 1 by Perron's theorem \cite{xiandai}.  
	Hence, $\frac{\dkh{\bar{\mathbf{W}}_{k}+\bar{\mathbf{W}}_{k}^{T}}}{2}$ is a symmetric stochastic matrix which has the largest eigenvalue 1, and the eigenvector associated with $1$ is $\mathbf{1}\in \mathbb{R}^{N}$. Now, for any nonzero column vector  $z\in \mathbf{R}^{N}$, 
	$$z^{T}\bar{\mathbf{W}}_{k}z=z^{T}\frac{\dkh{\bar{\mathbf{W}}_{k}+\bar{\mathbf{W}}_{k}^{T}}}{2}z\leq z^{T}z.$$  
	Moreover, for any nonzero $u\in\mathbf{R}^{N-1}$, 
	\begin{align}\label{eqn-W}
	&u\dkh{T_{1}\bar{\mathbf{W}}_{k}T_{1}^{T}-\mathbf{I}_{N-1}}u^{T}\nonumber\\
	=&\dkh{uT_{1}}\bar{\mathbf{W}}_{k}\dkh{uT_{1}}^{T}-\dkh{uT_{1}}\dkh{uT_{1}}^{T}\leq 0.
	\end{align}
	Similarly, 
	\begin{align}\label{eqn-Wt}
	&u\dkh{T_{1}\bar{\mathbf{W}}_{k}^{T}T_{1}^{T}-\mathbf{I}_{N-1}}u^{T}\nonumber\\
	=&\dkh{uT_{1}}\bar{\mathbf{W}}_{k}^{T}\dkh{uT_{1}}^{T}-\dkh{uT_{1}}\dkh{uT_{1}}^{T}\leq 0.
	\end{align}
	By $\dkh{\ref{eqn-W}}$--\eqref{eqn-Wt}, it is easy to obtain that 
	\begin{equation}
	\label{condition1}
	\mathbf{F}_{k}+\mathbf{F}_{k}^{T}\leq 0.
	\end{equation}
	\par
	Via the jointly connectivity of the network defined in $\mathbf{A2}$,   $\frac{1}{T+1}\sum_{s=t}^{t+T}\bar{\mathbf{W}}_{s}$ and  $\frac{1}{T+1}\sum_{s=t}^{t+T}\bar{\mathbf{W}}_{s}^{T}$  are irreducible double stochastic matrix.
	Then,  for any $z\in \mathbb{R}^{N}$,  it can obtain  that 
	\begin{equation*}
	z\dkh{\frac{1}{2\dkh{T+1}}\sum_{s=t}^{t+T}\dkh{\bar{\mathbf{W}}_{s}+\bar{\mathbf{W}}_{s}^{T}}}z^{T}-zz^{T}\leq 0,
	\end{equation*}
	where the equality holds if and only if 
	$z=c\mathbf{1}$.  Since $u^{T}T_{1}\mathbf{1}=0$, 
	$u^{T}T_{1}$ can  not be expressed as $c\mathbf{1}$ for any constant $c$. Consequently,  for any nonzero $u\in \mathbb{R}^{N-1}$, the following strict inequality must hold:
	\begin{align*}
	&\dkh{u^{T}T_{1}}\dkh{\frac{1}{2\dkh{T+1}}\sum_{s=t}^{t+T}\dkh{\bar{\mathbf{W}}_{s}+\bar{\mathbf{W}}_{s}^{T}}}{\dkh{u^{T}T_{1}}}^{T}\\
	<& \dkh{u^{T}T_{1}}{\dkh{u^{T}T_{1}}}^{T}.
	\end{align*}
	Notice $T_{1}T_{1}^{T}=\mathbf{I}_{N-1}$, which implies that for any nonzero $u\in \mathbb{R}^{N-1},$
	\begin{align*}
	\frac{1}{2\dkh{T+1}}\sum_{s=t}^{t+T}u^{T}\dkh{T_{1}\dkh{\bar{\mathbf{W}}_{s}+\bar{\mathbf{W}}_{s}^{T}}T_{1}^{T}-2I_{N-1}}u<0.
	\end{align*}
	As a result, $\sum_{s=t}^{t+T}\dkh{\mathbf{F}_{s}+\mathbf{F}_{s}^{T}}<0$. In addition, with the assumption on the uniform lower bound in $\mathbf{A1}$, there is a constant $\beta>0$ such that 
	\begin{equation}\label{condition2}
	\sum_{s=t}^{t+T}\dkh{\mathbf{F}_{s}+\mathbf{F}_{s}^{T}}\leq -\beta \mathbf{I}_{N-1}.
	\end{equation}
	By \eqref{condition1} and \eqref{condition2}, we have verified the conditions on linear matrix sequence $\hkh{H_{t}}$ in Lemma \ref{fang}. 
	\par 
	Now,  we analyze the noise term $\mathbf{C}_{k}$ and  $\mathbf{M}_{k}$ in the iteration $\dkh{\ref{stochastic iteration 2}}$. 
According to Lemma \ref{proof}, $\mathbf{M}_{k}$ is a martingale difference sequence, we obtain $\sum_{k=0}^{\infty}\delta_{k}\mathbf{M}_{k}<\infty, ~ a.s.$ via the martingale convergence theorem \cite{gailvlun}. 
By assumption $\mathbf{A3}$, on the chosen scheme of correct function $\mathbf{P}_{k}$, we have that	
\begin{align}\label{zaosheng1}
&\|\mathbf{C}_{k}\|^{2}\\
=&\dkh{\mathbf{P}_{k}-\mathbf{Q}_{k}}^{T}\dkh{\dkh{\bar{\mathbf{W}}_{k}-\mathbf{B}_{k}}^{T}\otimes\mathbf{I}_{M}}\nonumber\\
&\qquad\dkh{\dkh{\bar{\mathbf{W}}_{k}-\mathbf{B}_{k}}\otimes\mathbf{I}_{M}}\dkh{\mathbf{P}_{k}-\mathbf{Q}_{k}}\nonumber\\
=&\dkh{\mathbf{P}_{k}-\mathbf{Q}_{k}}^{T}\fkh{\dkh{\dkh{\bar{\mathbf{W}}_{k}-\mathbf{B}_{k}}^{T}\dkh{\bar{\mathbf{W}}_{k}-\mathbf{B}_{k}}}\otimes\mathbf{I}_{M}}\dkh{\mathbf{P}_{k}-\mathbf{Q}_{k}}\nonumber\\
\leq&\lambda_{max}\hkh{\dkh{\bar{\mathbf{W}}_{k}-\mathbf{B}_{k}}^{T}\dkh{\bar{\mathbf{W}}_{k}-\mathbf{B}_{k}}\otimes\mathbf{I}_{M}}\|\mathbf{P}_{k}-\mathbf{Q}_{k}\|^{2}\nonumber\\
&
\xrightarrow[k\rightarrow\infty]{}0,~~a.s..\nonumber
\end{align}
Consequently, $\lim\limits_{k\rightarrow \infty}\mathbf{C}_{k}=0,~a.s..$  
	\par
	In summary, we have verified all conditions in Lemma $\ref{fang}$, then we can obtain 
	\begin{equation*}
	\lim\limits_{k\rightarrow \infty }\dkh{T{\it \Gamma}\otimes\mathbf{I}_{M}}\mathbf{Q}_{k}=\lim\limits_{k\rightarrow \infty }\fkh{\xi_{k}^{T},0}^{T} =0,~~~a.s..
	\end{equation*}
	\indent Note that $T$ is an orthogonal matrix, so ${\it\Gamma}\otimes \mathbf{I}_{M} \mathbf{Q}_{k}\xrightarrow[k\rightarrow\infty]{}0$, a.s., where 
	\begin{align*}
	&\dkh{{\it\Gamma}\otimes \mathbf{I}_{M} }\mathbf{Q}_{k}\\
	=& \fkh{\dkh{\mathbf{I}_{N}-\frac{1}{N}\mathbf{1}\mathbf{1} ^{T}}\otimes \mathbf{I}_{M}}\mathbf{Q}_{k}=\mathbf{Q}_{k}-\dkh{\frac{1}{N} \mathbf{1}\mathbf{1}^{T}\otimes \mathbf{I}_{M}}\mathbf{Q}_{k}\\
	=&\mathbf{Q}_{k}-\dkh{\mathbf{1}\otimes\mathbf{I}_{M}}\dkh{\frac{1}{N}\mathbf{1}^{T}\otimes \mathbf{I}_{M}}\mathbf{Q}_{k}\\
	=&\mathbf{Q}_{k}-\dkh{\mathbf{1}\otimes\mathbf{I}_{M}}\dkh{\sum_{j=1}^{N}\frac{1}{N}Q_{j,k}},
	\end{align*}
	$i.e.$ $Q_{i,k}-\frac{1}{N}\sum_{j=1}^{N}Q_{j,k}\xrightarrow[k\rightarrow \infty]{ } 0, $~$a.s.$ for every agent $i\in \fkh{N}$.  Pre-multiplying update $\dkh{\ref{stochastic iteration 2}}$ by $\frac{1}{N}\mathbf{1}^{T}\otimes\mathbf{I}_{M}$ 
	yields 
	\begin{align}\label{sum}
	&\dkh{\frac{1}{N}\mathbf{1}^{T}\otimes\mathbf{I}_{M}}\mathbf{Q}_{k+1}\\
	=&\dkh{\frac{1}{N}\mathbf{1}^{T}\otimes\mathbf{I}_{M}}\mathbf{Q}_{k}+ \delta_{k}\bigg\{\dkh{\frac{1}{N}\mathbf{1}^{T}\otimes\mathbf{I}_{M}}\dkh{\bar{\mathbf{W}}_{k}-\mathbf{I}_{k}}\otimes\mathbf{I}_{M}\mathbf{Q}_{k}\nonumber\\
	&+\dkh{\frac{1}{N}\mathbf{1}^{T}\otimes\mathbf{I}_{M}}\dkh{\mathbf{C}_{k}+\mathbf{M}_{k}}\bigg\}\nonumber\\
	=&\dkh{\frac{1}{N}\mathbf{1}^{T}\otimes\mathbf{I}_{M}}\mathbf{Q}_{k}+ \delta_{k}\bigg\{\dkh{\frac{1}{N}\mathbf{1}^{T}\otimes\mathbf{I}_{M}}\dkh{\mathbf{C}_{k}+\mathbf{M}_{k}}\bigg\},\nonumber
	\end{align}
	because row sum of the matrix $\bar{\mathbf{W}}_{k}-\mathbf{I}_{N}$ is \emph{zero}. Summing equation $\dkh{\ref{sum}}$ from $k=0$ to $\infty$ yields 
	\begin{align*}
	&\lim\limits_{k\rightarrow \infty}\dkh{\frac{1}{N}\mathbf{1}^{T}\otimes\mathbf{I}_{M}}\mathbf{Q}_{k+1}\\
	=&\dkh{\frac{1}{N}\mathbf{1}^{T}\otimes\mathbf{I}_{M}}\mathbf{Q}_{0}+\dkh{\frac{1}{N}\mathbf{1}^{T}\otimes\mathbf{I}_{M}}\sum_{t=0}^{\infty}\delta_{k}\dkh{\mathbf{C}_{k}+\mathbf{M}_{k}}\\
	=&\dkh{\frac{1}{N}\mathbf{1}^{T}\otimes\mathbf{I}_{M}}\mathbf{Q}_{0}\\
	&+\frac{1}{N}\sum_{t=0}^{\infty}\delta_{k}\dkh{\mathbf{1}^{T}\dkh{\mathbf{W}_{k}-\mathbf{B}_{k}}\otimes\mathbf{I}_{M}}\dkh{\mathbf{Y}_{k}-\mathbf{Q}_{k}}\\
	\triangleq&\mathbf{q}^{\ast}.
	\end{align*} 	
	Therefore,  
	$Q_{i,k}\xrightarrow[k\rightarrow\infty]{}q^{\ast},~~a.s. ,~~\forall i\in \fkh{N}$.  
	\par
	In the following,  we will verify $\mathbf{q}^{\ast}<\infty$  if   $\|\mathbf{P}_{k}-\mathbf{Q}_{k}\|=O\dkh{\delta_{k}}$.   Denote
	\begin{align}
	L_{1}\triangleq&\dkh{\frac{1}{N}\mathbf{1}^{T}\otimes\mathbf{I}_{M}}\sum_{k=1}^{\infty}\delta_{k}\mathbf{C}_{k}\nonumber\\
	=&\frac{1}{N}\sum_{k=1}^{\infty}\delta_{k}\dkh{\mathbf{1}^{T}\dkh{\bar{\mathbf{W}}_{k}-\mathbf{B}_{k}}\otimes\mathbf{I}_{M}}\dkh{\mathbf{P}_{k}-\mathbf{Q}_{k}},\\
	L_{2}\triangleq&\dkh{\frac{1}{N}\mathbf{1}^{T}\otimes\mathbf{I}_{M}}\sum_{k=1}^{\infty}\delta_{k}\mathbf{M}_{k},		
	\end{align}
	then, due to $\|\mathbf{P}_{k}-\mathbf{Q}_{k}\|=O\dkh{\delta_{k}}$ and $\mathbf{A1}$, we have
	\begin{align}
	\|L_{1}\|=&\|\frac{1}{N}\sum_{k=1}^{\infty}\delta_{k}\dkh{\mathbf{1}^{T}\dkh{\bar{\mathbf{W}}_{k}-\mathbf{B}_{k}}\otimes\mathbf{I}_{M}}\dkh{\mathbf{P}_{k}-\mathbf{Q}_{k}}\|\nonumber\\
	\leq&\frac{1}{N}\sum_{k=1}^{\infty}\delta_{k}\|\dkh{\mathbf{1}^{T}\dkh{\bar{\mathbf{W}}_{k}-\mathbf{B}_{k}}\otimes\mathbf{I}_{M}}\|\cdot\|\dkh{\mathbf{P}_{k}-\mathbf{Q}_{k}}\|\nonumber\\
	\leq & \frac{M}{N} \sum_{k=1}^{\infty}\delta_{k}^{2}<\infty.
	\end{align}
	In addition, we have known that $\hkh{\mathbf{M}_{k},\mathcal{F}_{k}}$ is a martingale difference sequence \cite{gailvlun} in Lemma \ref{proof},  thus $\sum_{k=0}^{\infty}\delta_{k}\mathbf{M}_{k}<\infty$, $i.e. $ $L_{2}<\infty$. In conclusion, $q^{\ast}<\infty$. 
\end{proof}
\begin{tuilun}(\textbf{Strong consistency})
	Let $\mathbf{B}_{k}\equiv \mathbf{I}_{N}$ and $\mathbf{A1-A4}$ hold, then   $\hkh{\mathbf{Q}_{k}}$	 generated by $\dkh{\ref{stochastic iteration 2}}$ converges to the true distribution, $i.e.$ 
	\begin{align*}
	&\lim\limits_{k\rightarrow \infty }\dkh{Q_{i,k}-Q_{j,k}}=0, ~~\forall i,j \in \fkh{N}, ~a.s.,	\\
	&\lim\limits_{k \rightarrow \infty}Q_{i,k}= q^{\ast}\triangleq\dkh{\frac{1}{N}\mathbf{1}^{T}\otimes \mathbf{I}_{M}}\mathbf{Q}_{0}, ~~~a.s.. 
	\end{align*}
\end{tuilun}
\begin{zhushi}
	{\rm
		Compared with the undirected graph in \cite{6827923}, the joint connectivity of directed graph pointed at $\mathbf{A2}$ in  Theorem \ref{theorem} is a weaker condition. Besides, we have derived strong consensus to a finite limit $q^{\ast}$, which is almost identical with the true distribution ${\it \Pi}$.
		
	}	
\end{zhushi}
\section{Asymptotic Normality}\label{disijie}
In this section, we will establish asymptotic normality for estimate error $\mathbf{Q}_{k}-\mathbf{Q}^{\ast}$ of the distributed social sampling algorithm. The main tool for asymptotic  normality  analysis is shown in the following lemma.
\begin{yinli}[Theorem 3.3.1 in \cite{Chen1979Stochastic}]\label{chen2002}
	Let $H_{k}$ and $H$ be $l\times l$-matrices, $\hkh{x_{k}}$ be given by $x_{k+1}=x_{k}+\delta_{k}\dkh{H_{k}x_{k}+e_{k+1}+v_{k+1}}$ with an arbitrarily given initial value. Assume that the step-size $\delta_{k}$ satisfies $\mathbf{A1}$ and the following conditions hold
	\begin{itemize}
		\item[C1] $H_{k}\xrightarrow[k\rightarrow\infty]{}H$ and $H+\frac{\delta}{2}\mathbf{I}$ is stable with the constant $\delta$ given in $\mathbf{A1}$;
		\item[C2] $v_{k}=o\dkh{\sqrt{\delta_{k}}}$; 
		\item[C3]  $\hkh{e_{k},\mathcal{F}_{k}}$ is a martingale difference sequence of $l$-dimension which satisfies 
		\begin{align*}
		&E\dkh{e_{k}|\mathcal{F}_{k-1}}=0\nonumber \text{ and }\sup_{k}E\dkh{\|e_{k}\|^{2}|\mathcal{F}_{k-1}}\leq\sigma\nonumber\\
		& \text{with } \sigma \text{ being  a  constant},\\
		&\lim\limits_{k\rightarrow \infty}E\dkh{e_{k}e_{k}^{T}|\mathcal{F}_{k-1}}=\lim\limits_{k\rightarrow \infty}E e_{k}e_{k}^{T}\triangleq S_{0},~~a.s. ,\\
		&\lim\limits_{N\rightarrow \infty}\sup_{k}	E\|e_{k}\|^{2}\mathbf{I}_{\fkh{\|e_{k}\|>N}}=0,
		\end{align*}	
	\end{itemize}	
	then 	$\frac{x_{k}}{\sqrt{\delta_{k}}}\xrightarrow[k\rightarrow\infty]{d} N\dkh{0,S},$
	where 
	\begin{align*}
S=\int_{0}^{\infty} e^{\dkh{H+\frac{\delta}{2}\mathbf{I}}t}S_{0}e^{\dkh{H^{T}+\frac{\delta}{2}\mathbf{I}}t} dt.
	\end{align*}
\end{yinli}
For the case where the root set of the observation function $f\dkh{x}=\bar{\mathbf{H}}_{k}\dkh{x}$ consists of a singleton $zero$, we consider $\hkh{\xi_{k}}$ which is defined in $\dkh{\ref{jiangwei}}$. It has been verified  that $\lim\limits_{k\rightarrow \infty }\xi_{k}=0, ~a.s.$. 
Rewrite \eqref{jiangwei} as
\begin{align}\label{zhengtai}
\xi_{k+1}=&\xi_{k}+ \delta_{k}\{\dkh{T_{1}\otimes \mathbf{I}_{M}}\bar{\mathbf{H}}_{k}\dkh{T_{1}^{T}\otimes\mathbf{I}_{M}}\xi_{k}\nonumber\\
	&+\dkh{T_{1}\otimes\mathbf{I}_{M}}\dkh{\mathbf{C}_{k}+\mathbf{M}_{k}}\},
\end{align} 
where $\bar{\mathbf{H}}_{k}$, $\mathbf{C}_{k}$ and $\mathbf{M}_{k}$ are given by $\dkh{\ref{yangchalie}}$. It can be seen that $\xi_{k}$  is updated by a linear stochastic approximation algorithm to approach the sought root $zero$. Furthermore, we can investigate the asymptotic properties of $\dkh{\ref{zhengtai}}$.\\
\indent Before describing the convergent rate of iteration $\dkh{\ref{zhengtai}}$, we require the following assumptions.  We keep $\mathbf{A1}$ unchanged, but strengthen $\mathbf{A3}$ to $\mathbf{A3^{'}}$ and change $\mathbf{A2}$ to $\mathbf{A2^{'}}$ as follows. 
\begin{description}
	\item[$\mathbf{A2^{'}}$]  
	\begin{enumerate}
		\item[$\dkh{i}$] $\hkh{\mathbf{W}_{k}}$ is an $i.i.d.$ sequence.
		\item[$\dkh{ii}$] The mean graph $\bar{\mathcal{G}}_{k}=\dkh{\mathcal{V},\mathcal{E}_{k},\bar{\mathbf{W}}}$ generated by $\bar{\mathbf{W}}=E\fkh{\mathbf{W}_{k}}$ is strongly connected and the adjacency matrix  $\mathbf{W}_{k}$ is double stochastic.
	\end{enumerate}
	\item[$\mathbf{A3^{'}}$] $\|\mathbf{P}_{k}-\mathbf{Q}_{k}\|=o\dkh{\delta_{k}}$.	
	\item[$\mathbf{A4^{'}}$] Choose $\mathbf{B}_{k}\equiv\mathbf{B}$, where $\mathbf{B}$ is a constant matrix.	
	\item[$\mathbf{A5}$] For sampling $\mathbf{Y}_k$ under distribution $\mathbf{P}_k$, we assume that $\varSigma$  is a constant matrix almost surely, where $\varSigma\triangleq \lim\limits_{k\rightarrow\infty}E\fkh{\dkh{\mathbf{Y}_{k}-\mathbf{P}_{k}}\dkh{\mathbf{Y}_{k}-\mathbf{P}_{k}}^{T}|\mathcal{F}_{k-1}}$.
\end{description}

\begin{zhushi}
	{\rm We consider an example of $\mathbf{A5}$ for illustration.  Recall that the random message $\mathbf{Y}_{k}\triangleq \dkh{Y_{1,k}^{T},Y_{2,k}^{T},\cdots,Y_{N,k}^{T}}^{T}$ is generated from the sampling function $\mathbf{P}_{k}\triangleq \dkh{P_{1,k}^{T},P_{2,k}^{T},\cdots,P_{N,k}^{T}}^{T}$. Let $Y_{i,k}=\mathbf{e}_{m}$
		with probability $P_{i,k}^{m}$, where $m\in\fkh{M}$. Then 
		\begin{align}
		&\varSigma_{k}\triangleq E\fkh{\dkh{\mathbf{Y}_{k}-\mathbf{P}_{k}}\dkh{\mathbf{Y}_{k}-\mathbf{P}_{k}}^{T}|\mathcal{F}_{k-1}}\nonumber\\
		=&
		\begin{pmatrix}
		\varSigma_{1,k}& 0&\cdots&0\\
		0&\varSigma_{2,k}&\cdots&0\\
		\vdots & \vdots & \ddots &\vdots \\
		0&0&\cdots&\varSigma_{N,k}\\
		\end{pmatrix},
		\end{align} 		
		where
		\begin{align*}
		&\varSigma_{i,k}=E\fkh{\dkh{Y_{i,k}-P_{i,k}}\dkh{Y_{i,k}-P_{i,k}}^{T}|\mathcal{F}_{k-1}}\\
		=&
		\begin{pmatrix}
		\dkh{1-P_{i,k}^{1}}P_{i,k}^{1}&\cdots&0\\
		\vdots&\ddots&\vdots\\
		0&\cdots&\dkh{1-P_{i,k}^{M}}P_{i,k}^{M}
		\end{pmatrix}.
		\end{align*}
		From $\mathbf{P}_{k}-\mathbf{Q}_{k}\xrightarrow[k\rightarrow\infty]{} 0~a.s.$, $\mathbf{Q}_{k}-\mathbf{Q}^{\ast}\xrightarrow[k\rightarrow\infty]{} 0~a.s.$, we have 
		\begin{align*}
		\varSigma_{k}\xrightarrow[k\rightarrow\infty]{}\varSigma=\hkh{\dkh{\mathbf{I}_{M}-diag\dkh{q^{\ast}}}diag\dkh{q^{\ast}}}\otimes\mathbf{I}_{N}.
		\end{align*}
	}	
\end{zhushi}
\indent The following lemma considers the martingale difference sequence part. 
\begin{yinli}\label{yangchana}
	Under $\mathbf{A2^{'}}$ and $\mathbf{A4^{'}}$, by choosing  $\mathbf{B}\equiv\mathbf{I}_{N}$, we have
	\begin{align}\label{yangchazaosheng}
	&\varepsilon_{k}\triangleq\dkh{T_{1}\otimes\mathbf{I}_{M}}\mathbf{M}_{k}\nonumber\\
	=&\dkh{T_{1}\dkh{\mathbf{W}_{k}-\mathbf{B}}\otimes\mathbf{I}_{M}}\dkh{\mathbf{Y}_{k}-\mathbf{P}_{k}}\nonumber\\
	&+\dkh{T_{1}\dkh{\mathbf{W}_{k}-\bar{\mathbf{W}}}\otimes\mathbf{I}_{M}}\mathbf{P}_{k}
	\end{align}		
	is a martingale difference sequence satisfying
	\begin{align}
	&E\dkh{\varepsilon_{k}|\mathcal{F}_{k-1}}=0,\nonumber ~ \\ &\sup_{k}E\dkh{\|\varepsilon_{k}\|^{2}|\mathcal{F}_{k-1}}\leq\sigma~\text{with $\sigma$  being  a  constant,}\label{yangchayoujie}\\
	&\lim\limits_{N\rightarrow \infty}\sup_{k}	E\|\varepsilon_{k}\|^{2}\mathbf{I}_{\fkh{\|\varepsilon_{k}\|>N}}=0\label{jixian}.
	\end{align}
\end{yinli}
\begin{proof}
	See the proof in Appendix B. 
\end{proof}
Now we can establish the asymptotic normality of the  distributed social sampling algorithm   $\dkh{\ref{zhengtai}}$.
\begin{dingli}[\textbf{Asymptotic normality}]\label{zhengtaijieguo}
	Let $\mathbf{A1}$, $\mathbf{A2^{'}}$, $\mathbf{A3^{'}}$, $\mathbf{A4^{'}}$ and $\mathbf{A5}$ hold, then  $\xi_{k}= \dkh{T_{1}\otimes \mathbf{I}_{M}} \mathbf{Q}_{k}$ is asymptotically normal,  $i.e.$ the distribution of $\frac{1}{\sqrt{\delta_{k}}}\xi_{k}$ converge to a normal distribution:
	\begin{equation}
	\frac{\xi_{k}}{\sqrt{\delta_{k}}}\xrightarrow[k\rightarrow\infty]{d} N\dkh{\mathbf{0},\mathbf{S}},
	\end{equation}
	where
	\begin{align*}
	&\mathbf{S}=\int_{0}^{\infty}e^{\dkh{\bar{\mathbf{F}}+\frac{\delta}{2}\mathbf{I}_{\dkh{N-1}M}}t}\mathbf{S}_{0}e^{\dkh{\bar{\mathbf{F}}+\frac{\delta}{2}\mathbf{I}_{\dkh{N-1}M}}^{T}t}dt,\\
	&\mathbf{S}_{0}\triangleq E\fkh{\dkh{T_{1}\dkh{\mathbf{W}_{k}-\mathbf{I}_{N}}\otimes\mathbf{I}_{M}}}\varSigma E\fkh{\dkh{\mathbf{W}_{k}-\mathbf{I}_{N}}^{T}T_{1}^{T}\otimes\mathbf{I}_{M}},\\
	&\bar{\mathbf{F}}\triangleq\dkh{T_{1}\dkh{\bar{\mathbf{W}}-\mathbf{I}_{N}}T_{1}^{T}}\otimes\mathbf{I}_{M}, \\
	&\varSigma\triangleq \lim\limits_{k\rightarrow\infty}E\fkh{\dkh{\mathbf{Y}_{k}-\mathbf{P}_{k}}\dkh{\mathbf{Y}_{k}-\mathbf{P}_{k}}^{T}|\mathcal{F}_{k-1}}.
	\end{align*}
\end{dingli}

\begin{proof}
	To use Lemma $\ref{chen2002}$,  we have to validate conditions C1 and C2.  
	\par
	First,  we consider C1. Denote $\bar{\mathbf{F}}_{k}\triangleq \dkh{T_{1}\otimes \mathbf{I}_{M}}\bar{\mathbf{H}}_{k}\dkh{T_{1}^{T}\otimes \mathbf{I}_{M}}$, by $\mathbf{A2^{'}}$ on the weighted matrix $\mathbf{W}_{k}$, we have
	\begin{align*}
	&\bar{\mathbf{F}}+\frac{\delta}{2}\mathbf{I}_{\dkh{N-1}M}\\
	=&\dkh{T_{1}\otimes\mathbf{I}_{M}}\dkh{\dkh{\bar{\mathbf{W}}-\mathbf{I}_{N}}\otimes\mathbf{I}_{M}}\dkh{T_{1}^{T}\otimes\mathbf{I}_{M}}+\frac{\delta}{2}\mathbf{I}_{\dkh{N-1}M}\\
	=&\dkh{T_{1}\dkh{\bar{\mathbf{W}}-\mathbf{I}_{N}}T_{1}^{T}}\otimes\mathbf{I}_{M}+\frac{\delta}{2}\mathbf{I}_{\dkh{N-1}M}.
	\end{align*}	
	Under  assumption $\mathbf{A2^{'}}$, $\bar{\mathbf{W}}=E\fkh{\mathbf{W}_{k}}$ is a stochastic matrix and its largest eigenvalue is 1 via Perron's Theorem \cite{xiandai}. Let  the second largest eigenvalue of $\bar{\mathbf{W}}$ be  $\lambda_{2}$. We can choose step-size $\delta_{k}$ properly such that $\lambda_{2}<1-\frac{\delta}{2}$, where the linear matrix in $\dkh{\ref{zhengtai}}$ satisfies the stable  assumption in C1.\\	
	\indent Now we analyze C2 item by item. First $v_{k}\triangleq\dkh{T_{1}\otimes\mathbf{I}_{M}}\mathbf{C}_{k}$, where $\mathbf{C}_{k}$ is defined in $\dkh{\ref{yangchalie}}$. We can obtain $v_{k}=o\dkh{\delta_{k}}$ according to $\mathbf{A3^{'}}$ and $\dkh{\ref{zaosheng1}}$.  According to $\dkh{\ref{yangchazaosheng}}$, the martingale difference part of noise has
	\begin{align*}
	&\varepsilon_{k}\varepsilon_{k}^{T}\\
	=&\dkh{T_{1}\dkh{\mathbf{W}_{k}-\bar{\mathbf{W}}}\otimes \mathbf{I}_{M}}\mathbf{P}_{k}\mathbf{P}_{k}^{T}\dkh{\dkh{\mathbf{W}_{k}-\bar{\mathbf{W}}}^{T}T_{1}^{T}\otimes\mathbf{I}_{M}}\\
	&+\dkh{T_{1}\dkh{\mathbf{W}_{k}-\bar{\mathbf{W}}}\otimes \mathbf{I}_{M}}\mathbf{P}_{k}\dkh{\mathbf{Y}_{k}-\mathbf{P}_{k}}^{T}\\
	&\qquad\cdot\dkh{\dkh{\mathbf{W}_{k}-\mathbf{I}_{N}}^{T}T_{1}^{T}\otimes\mathbf{I}_{M}}\\
	&+\dkh{T_{1}\dkh{\mathbf{W}_{k}-\mathbf{I}_{N}}\otimes\mathbf{I}_{M}}\dkh{\mathbf{Y}_{k}-\mathbf{P}_{k}}\mathbf{P}_{k}^{T}\\
&	\qquad\cdot\dkh{\dkh{\mathbf{W}_{k}-\bar{\mathbf{W}}}^{T}T_{1}^{T}\otimes\mathbf{I}_{M}}\\
	&+\dkh{T_{1}\dkh{\mathbf{W}_{k}-\mathbf{I}_{N}}\otimes\mathbf{I}_{M}}\dkh{\mathbf{Y}_{k}-\mathbf{P}_{k}}\\
	&\qquad\cdot\dkh{\mathbf{Y}_{k}-\mathbf{P}_{k}}^{T}\dkh{\dkh{\mathbf{W}_{k}-\mathbf{I}_{N}}^{T}T_{1}^{T}\otimes\mathbf{I}_{M}}\\
	\triangleq&S_{1}+S_{2}+S_{3}+S_{4}.
	\end{align*}
	Notice that 
	\begin{align*}
	&\dkh{T_{1}\dkh{\mathbf{W}_{k}-\bar{\mathbf{W}}}\otimes\mathbf{I}_{M}}\mathbf{P}_{k}\\
	=&\dkh{T_{1}\dkh{\mathbf{W}_{k}-\bar{\mathbf{W}}}\otimes \mathbf{I}_{M}}\dkh{\mathbf{P}_{k}-\mathbf{Q}_{k}}\\
	&+\dkh{T_{1}\dkh{\mathbf{W}_{k}-\bar{\mathbf{W}}}\otimes \mathbf{I}_{M}}\dkh{\mathbf{Q}_{k}-\mathbf{Q}^{\ast}}\\
	&+\dkh{T_{1}\dkh{\mathbf{W}_{k}-\bar{\mathbf{W}}}\otimes \mathbf{I}_{M}}\mathbf{Q}^{\ast}\\
	=&\dkh{T_{1}\dkh{\mathbf{W}_{k}-\bar{\mathbf{W}}}\otimes \mathbf{I}_{M}}\dkh{\mathbf{P}_{k}-\mathbf{Q}_{k}}\\
	&+\dkh{T_{1}\dkh{\mathbf{W}_{k}-\bar{\mathbf{W}}}\otimes \mathbf{I}_{M}}\dkh{\mathbf{Q}_{k}-\mathbf{Q}^{\ast}}\\
	&+\dkh{T_{1}\dkh{\mathbf{W}_{k}-\bar{\mathbf{W}}}\otimes \mathbf{I}_{M}}\dkh{\mathbf{1}\otimes\dkh{\frac{1}{N}\mathbf{1}^{T}\otimes\mathbf{I}_{M}\mathbf{Q}_{0}}}\\
	=&\dkh{T_{1}\dkh{\mathbf{W}_{k}-\bar{\mathbf{W}}}\otimes \mathbf{I}_{M}}\dkh{\mathbf{P}_{k}-\mathbf{Q}_{k}}\\
	&+\dkh{T_{1}\dkh{\mathbf{W}_{k}-\bar{\mathbf{W}}}\otimes \mathbf{I}_{M}}\dkh{\mathbf{Q}_{k}-\mathbf{Q}^{\ast}}\\
	&+\dkh{T_{1}\dkh{\mathbf{W}_{k}-\bar{\mathbf{W}}}\mathbf{1}}\otimes\dkh{\frac{1}{N}\mathbf{I}_{M}\dkh{\mathbf{1}^{T}\otimes\mathbf{I}_{M}}\mathbf{Q}_{0}}.
	\end{align*}
	Since $\mathbf{W}_{k}$ and $\bar{\mathbf{W}}$ all are double stochastic matrix and $\|\mathbf{P}_{k}-\mathbf{Q}_{k}\|\xrightarrow[k\rightarrow\infty]{}0,~a.s.$,   $\mathbf{Q}_{k}-\mathbf{Q}^{\ast}\xrightarrow[k\rightarrow\infty]{}0,~a.s.$, we have 
	\begin{align}
	\dkh{T_{1}\dkh{\mathbf{W}_{k}-\bar{\mathbf{W}}}\otimes\mathbf{I}_{M}}\mathbf{P}_{k}\xrightarrow[k\rightarrow\infty]{}0,~~a.s..
	\end{align}
	Therefore, we can obtain $\lim\limits_{k\rightarrow\infty}E\fkh{S_{1}|\mathcal{F}_{k-1}}=0,~a.s.$,
	$\lim\limits_{k\rightarrow\infty}E\fkh{S_{2}|\mathcal{F}_{k-1}}=0,~a.s.$ and $\lim\limits_{k\rightarrow\infty}E\fkh{S_{3}|\mathcal{F}_{k-1}}=0,~a.s.$. 
	\par
	On the other hand, due to $\mathbf{A5}$, $\varSigma\triangleq \lim\limits_{k\rightarrow\infty}E\fkh{\dkh{\mathbf{Y}_{k}-\mathbf{P}_{k}}\dkh{\mathbf{Y}_{k}-\mathbf{P}_{k}}^{T}|\mathcal{F}_{k-1}}$  is a constant matrix  almost surely.   Besides,  $\mathbf{W}_{k}$ is independent with $\mathcal{F}_{k-1}$ in $\mathbf{A2^{'}}$, then
	\begin{align*}
&\lim\limits_{k\rightarrow\infty}E\{S_{4}|\mathcal{F}_{k-1}\}\\
=&\lim\limits_{k\rightarrow\infty}E\{\dkh{T_{1}\dkh{\mathbf{W}_{k}-\mathbf{I}_{N}}\otimes\mathbf{I}_{M}}\\
&\dkh{\mathbf{Y}_{k}-\mathbf{P}_{k}}\dkh{\mathbf{Y}_{k}-\mathbf{P}_{k}}^{T}\dkh{\dkh{\mathbf{W}_{k}-\mathbf{I}_{N}}^{T}T_{1}^{T}\otimes\mathbf{I}_{M}}|\mathcal{F}_{k-1}\}\\
=&\lim\limits_{k\rightarrow\infty}E\{E\{\dkh{T_{1}\dkh{\mathbf{W}_{k}-\mathbf{I}_{N}}\otimes\mathbf{I}_{M}}\dkh{\mathbf{Y}_{k}-\mathbf{P}_{k}}\dkh{\mathbf{Y}_{k}-\mathbf{P}_{k}}^{T}\\
&\dkh{\dkh{\mathbf{W}_{k}-\mathbf{I}_{N}}^{T}T_{1}^{T}\otimes\mathbf{I}_{M}}|\mathcal{F}_{k-1},\mathbf{W}_{k}\}|\mathcal{F}_{k-1}\}\\
=&\lim\limits_{k\rightarrow\infty}E\{\dkh{T_{1}\dkh{\mathbf{W}_{k}-\mathbf{I}_{N}}\otimes\mathbf{I}_{M}}\\
&E\{\dkh{\mathbf{Y}_{k}-\mathbf{P}_{k}}\dkh{\mathbf{Y}_{k}-\mathbf{P}_{k}}^{T}|\mathcal{F}_{k-1},\mathbf{W}_{k}\}\\
&\dkh{\dkh{\mathbf{W}_{k}-\mathbf{I}_{N}}^{T}T_{1}^{T}\otimes\mathbf{I}_{M}}|\mathcal{F}_{k-1}\}\\
=&\lim\limits_{k\rightarrow\infty}E\{\dkh{T_{1}\dkh{\mathbf{W}_{k}-\mathbf{I}_{N}}\otimes\mathbf{I}_{M}}|\mathcal{F}_{k-1}\}\\
&E\{E\{\dkh{\mathbf{Y}_{k}-\mathbf{P}_{k}}\dkh{\mathbf{Y}_{k}-\mathbf{P}_{k}}^{T}|\mathcal{F}_{k-1},\mathbf{W}_{k}\}|\mathcal{F}_{k-1}\}\\
&\times E\{\dkh{\mathbf{W}_{k}-\mathbf{I}_{N}}^{T}T_{1}^{T}\otimes\mathbf{I}_{M}|\mathcal{F}_{k-1}\}\\
=&E\{\dkh{T_{1}\dkh{\mathbf{W}_{k}-\mathbf{I}_{N}}\otimes\mathbf{I}_{M}}\}\cdot\varSigma\cdot E\{\dkh{\mathbf{W}_{k}-\mathbf{I}_{N}}^{T}T_{1}^{T}\otimes\mathbf{I}_{M}\}\triangleq\mathbf{S}_{0}.
\end{align*}
	According to  Lemma \ref{yangchana}, we have verified C2. 
	\par
	In summary, we have verified all conditions in Lemma \ref{chen2002},  thus  conclusion of this theorem holds.	
\end{proof}

\begin{tuilun}\label{zhengtaituilun}
	Suppose $\mathbf{A1}$, $\mathbf{A2^{'}}$, $\mathbf{A3^{'}}$, $\mathbf{A4^{'}}$ and $\mathbf{A5}$ hold, then 	
	\begin{equation*}
	\frac{\mathbf{Q}_{k}-\mathbf{Q	}^{\ast}}{\sqrt{\delta_{k}}}\xrightarrow[k\rightarrow\infty]{d} N\dkh{\mathbf{0},\widetilde{{\mathbf{S}}}},
	\end{equation*}
	where $\widetilde{{\mathbf{S}}}\triangleq \dkh{T{\it\Gamma}\otimes \mathbf{I}_{M}}^{T}
	\begin{pmatrix}
	\mathbf{S}&0\\
	0&0\\
	\end{pmatrix}
	\dkh{T{\it\Gamma}\otimes \mathbf{I}_{M}}.$
\end{tuilun}

\begin{zhushi}
	Theorem \ref{zhengtaijieguo} and Corollary \ref{zhengtaituilun} establish that the error between estimates generated by algorithm $\dkh{\ref{stochastic iteration}}$ and true empirical distribution is asymptotically normal, and the asymptotic covariance is characterized by network topology and quantized protocol. Our analysis results are more detailed and profound than that in \cite{6827923}, which only gives bounds on the expected squared error.	
\end{zhushi}

\section{Numerical Simulation}\label{diwujie}
In this section we provide a numerical simulation for the distribution social sampling algorithm considered in $\dkh{\ref{linear update}}$. 
Let $N=50$ with the underlying graph being fully connected. Each agent holds an initial opinion $Q_{i,0}=\mathbf{e}_{X_{i}}$, which is drawn $i.i.d.$ from $\fkh{0.2~0.3~0.4~0.1}$. It means that the dimension of opinion state space  $M=4$. At each time $k$, agent $i$ generates its random message $Y_{i,k}\in\hkh{\mathbf{e}_{1},\mathbf{e}_{2},\mathbf{e}_{3},\mathbf{e}_{4}}$ based on the internal estimate $Q_{i,k}$ directly, $i.e.$, we choose $P_{i,k}=Q_{i,k}$ and do not make corrections. Setting  the mixed coefficients $a_{ii}^{k}=0$, $b_{ii}^{k}=1$, step size   $\delta_{k}=\frac{1}{k^{0.75}}$, we update the internal estimate sequence $\hkh{Q_{i,k}}$ according to iteration $\dkh{\ref{linear update}}$.
\begin{center}
	{\includegraphics[scale=0.5]{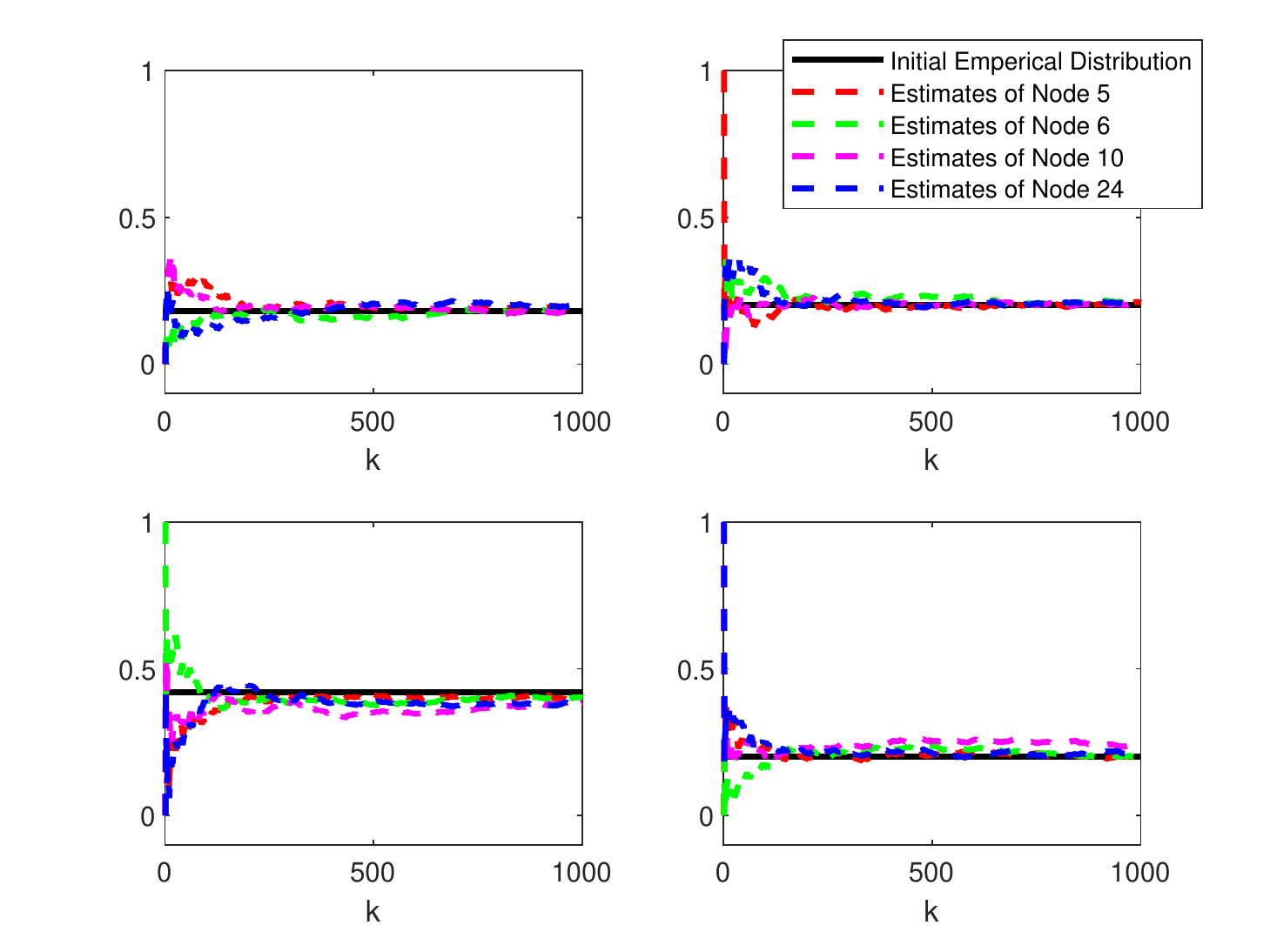}}\vskip3mm
	\centering{\small {\bf Figure 1:} Trace of estimate of $Q_{i,k}$ for $i\in \fkh{N}$ over every single opinion state  $m\in\fkh{M}$ with $M=4$ \label{fig1}}
\end{center}
\begin{center}
	{\includegraphics[scale=0.5]{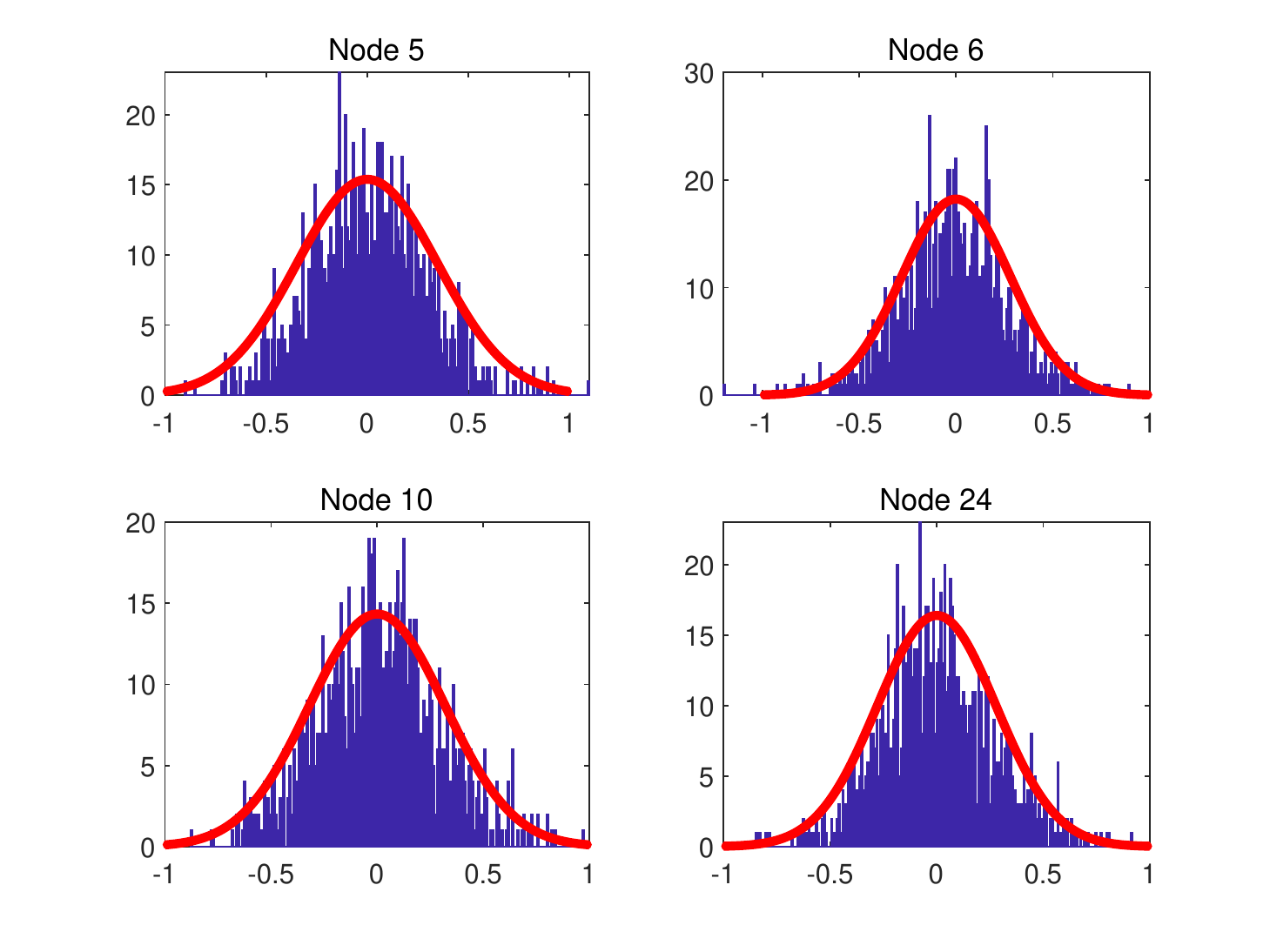}}\vskip3mm
	\centering{\small {\bf Figure 2:}  Histogram and limit distribution for $\dkh{Q_{i,k}-q^{\ast}}/\sqrt{\delta_{k}}$ at $k=500$.}\label{fig2}
\end{center}
\par
The trace of estimate sequence of empirical distribution $\Pi$ for some selected agents is shown in Figure 1, where each subgraph presents a state in $\fkh{M}=4$. As shown in Theorem \ref{theorem}, the estimated sequence generated by social sampling procedure converges to the true empirical distribution $q^{\ast}\triangleq\dkh{\frac{1}{N}\mathbf{1}^{T}\otimes \mathbf{I}_{M}}\mathbf{Q}_{0}$. We  have calculated the algorithm $\dkh{\ref{linear update}}$ for 1000 times independently. The histograms for each component of $\dkh{Q_{i,k}-q^{\ast}}/\sqrt{\delta_{k}}$ at $k=500$ are shown in Figure 2. It is shown that the data fits the normal distribution well.

\section{Conclusions}\label{diliujie}
In this work, convergence of distributed social sampling algorithm toward a common distribution has been established over random networks based on stochastic approximation.  We have proved that the distribution estimates derived by agents' local interaction reached consensus almost sure to a value, which is related with the true empirical distribution and accumulation of quantized error. Furthermore, the error between estimates and true empirical distribution has been shown to be asymptotically normal with zero mean and known covariance, which is characterized by network topology and the social sampling protocol.\\
\indent In fact, the randomized sample procedure is fairly general to be used in other problems, such as distributed optimization over large data sets. As the messages are quantized as identical vectors,  the computation complexity is significantly reduced. In the future work, we will dig deeper about this random message passing protocol.





\begin{appendix}
	\subsection*{Appendix A \qquad  Proof of Lemma \ref{proof}}\label{app_proof}
	According to the protocol of social sampling,  we  have $	E\fkh{Y_{i,k}}=\sum_{m=1}^{M} \mathbb{P}\dkh{Y_{i,k}=\mathbf{e}_{m}}
	=\sum_{m=1}^{M}\mathbf{e}_{m}P_{i,k}^{m}
	=P_{i,k}.$
	Thus, $E\fkh{\mathbf{Y}_{k}}=\mathbf{P}_{k}$.
	Taking conditional expectation given $\mathcal{F}_{k-1}$ over both sides of $\dkh{\ref{yangchalie}}$,  
	we obtain
	\begin{align*}
	E\fkh{\mathbf{M}_{k}|\mathcal{F}_{k-1}}=&E\fkh{\dkh{\dkh{\mathbf{W}_{k}-\mathbf{B}_{k}}\otimes\mathbf{I}_{M}}\dkh{\mathbf{Y}_{k}-\mathbf{P}_{k}}|\mathcal{F}_{k-1}}\\
	&+E\fkh{\dkh{\dkh{\mathbf{W}_{k}-\bar{\mathbf{W}}_{k}}\otimes\mathbf{I}_{M}}\mathbf{P}_{k}|\mathcal{F}_{k-1}}\\
	=& E\fkh{\dkh{\mathbf{W}_{k}-\mathbf{B}_{k}}\otimes\mathbf{I}_{M}} E\fkh{\dkh{\mathbf{Y}_{k}-\mathbf{P}_{k}}|\mathcal{F}_{k-1}}\\
	&+E\fkh{\dkh{\mathbf{W}_{k}-\bar{\mathbf{W}}_{k}}\otimes\mathbf{I}_{M}}\mathbf{P}_{k}=0.
	\end{align*}
	Hence, $\hkh{\mathbf{M}_{k},\mathcal{F}_{k}}$ is a martingale difference sequence. 
	
	\subsection*{Appendix B \qquad Proof of Lemma \ref{yangchana}}\label{app_yangchana}
	Since  we have verified that $\dkh{\mathbf{M}_{k},\mathcal{F}_{k}}$ is a martingale difference sequence and $T_{1}$ is a constant matrix, $\dkh{\varepsilon_{k},\mathcal{F}_{k}}$ given by $\dkh{\ref{yangchazaosheng}}$ is also a martingale difference sequence. At first,  we demonstrate the boundedness of  $\varepsilon_{k}$.  The random message $Y_{i,k}\in \mathcal{Y}=\hkh{\mathbf{0},\mathbf{e}_{1},\cdots, \mathbf{e}_{M}}$ of agent $i$ at time $k$ is  generated according to the $M$-dimension row probability vector $P_{i,k}\in \mathbb{P}\dkh{\mathcal{Y}}$. Besides, $\mathbf{W}_{k}$ is a double stochastic matrix by $\mathbf{A2^{'}}$,  then  we obtain 
	\begin{align*}
	\|\varepsilon_{k}\|\leq&\|\dkh{T_{1}\otimes \mathbf{I}_{M}}\dkh{\dkh{\mathbf{W}_{k}-\mathbf{I}_{N}}\otimes\mathbf{I}_{M}}\dkh{\mathbf{Y}_{k}-\mathbf{P}_{k}}\|\\
	&+\|\dkh{T_{1}\otimes\mathbf{I}_{M}}\dkh{\dkh{\mathbf{W}_{k}-\bar{\mathbf{W}}}\otimes\mathbf{I}_{M}}\mathbf{P}_{k}\|\\
	\leq& \|T_{1}\otimes\mathbf{I}_{M}\|\cdot\|\dkh{\mathbf{W}_{k}-\mathbf{I}_{N}}\otimes\mathbf{I}_{M}\|\cdot\|\mathbf{Y}_{k}-\mathbf{P}_{k}\|\\
	&+\|T_{1}\otimes\mathbf{I}_{M}\|\cdot\|\dkh{\mathbf{W}_{k}-\bar{\mathbf{W}}}\otimes\mathbf{I}_{M}\|\cdot\|\mathbf{P}_{k}\|
	\\
	&\leq cN^{2}M^{2}.
	\end{align*}	
	Hence, the noise sequence $\varepsilon_{k}$ is bounded. We can derive $\dkh{\ref{yangchayoujie}}$ and $\dkh{\ref{jixian}}$ directly.

\end{appendix}
	
\end{document}